\pdfoutput=1
\documentclass[10pt,twoside,twocolumn,letterpaper]{article}

\usepackage{times}
\usepackage[left=49.14pt,right=49.14pt,top=30.27pt,includehead,
            headheight=6pt,headsep=18pt,textheight=696pt,
            columnsep=12pt]{geometry}
\usepackage{xcolor}
\usepackage{fancyhdr}
\usepackage{caption}
\usepackage{cite}
\usepackage{amsmath,amssymb,amsfonts,amscd}
\usepackage{graphicx}
\usepackage{hyperref}
\usepackage{microtype}
\usepackage[all]{nowidow}
\usepackage{tikz}
\usetikzlibrary{arrows.meta,decorations.pathreplacing}

\definecolor{nblue}{rgb}{0,0.263,0.576}
\definecolor{subsectioncolor}{rgb}{0,0.541,0.855}

\makeatletter

\renewcommand\@maketitle{%
  \newpage
  \begin{center}%
    \vskip 0.9em%
    {\color{nblue}\sf\fontsize{24}{28}\selectfont\@title\par}%
    \vskip 1.0em\par
    {\lineskip .5em\sf\fontsize{11}{13}\selectfont\@author\par}%
  \end{center}%
  \par\vskip 1.5em}

\def\maketitle{\par
  \begingroup
    \normalfont
    \def\thefootnote{}%
    \def\footnotemark{}%
    \def\@makefnmark{}%
    \long\def\@makefntext##1{\sffamily\parindent 1em\indent##1}%
    \footnotesep 0.7\baselineskip
    \twocolumn[\@maketitle]%
    \thispagestyle{fancy}\@thanks
  \endgroup
  \setcounter{footnote}{0}%
  \global\let\thanks\relax
  \global\let\maketitle\relax
  \global\let\@maketitle\relax
  \global\let\@thanks\@empty}

\renewenvironment{abstract}{%
  \everymath={\sf}\sf\small\bfseries
  \color{subsectioncolor}\textit{Abstract}---\,\normalcolor\ignorespaces}%
  {\par\vspace{1.34ex}\normalfont\normalsize}

\newenvironment{IEEEkeywords}{%
  \everymath={\sf}\sf\small\bfseries
  \color{subsectioncolor}\textit{Index Terms}---\,\normalcolor\ignorespaces}%
  {\par\vspace{0.67ex}\normalfont\normalsize}

\renewcommand\@seccntformat[1]{\csname the#1\endcsname.\hskip 0.5em\relax}

\renewcommand\section{\@startsection{section}{1}{\z@}%
  {3.0ex plus 1.5ex minus 1.5ex}{0.7ex plus 1ex minus 0ex}%
  {\color{nblue}\centering
   \fontencoding{T1}\fontfamily{phv}\fontseries{m}\fontshape{n}%
   \fontsize{10}{12}\selectfont\scshape}}
\renewcommand\subsection{\@startsection{subsection}{2}{\z@}%
  {3.5ex plus 1.5ex minus 1.5ex}{0.7ex plus .5ex minus 0ex}%
  {\color{subsectioncolor}\normalsize\normalfont\everymath={\sf}\sf\itshape\raggedright}}
\renewcommand\subsubsection{\@startsection{subsubsection}{3}{\parindent}%
  {0ex plus 0.1ex minus 0.1ex}{0ex}%
  {\color{nblue}\normalfont\fontsize{9}{11}\selectfont\everymath={\sf}\sf\itshape}}

\newdimen\@thmtmpitemindent
\def\@begintheorem#1#2{\@thmtmpitemindent\itemindent\topsep 0pt\rmfamily\trivlist
  \item[\hskip\labelsep{\indent\itshape #1\ #2:}]\itemindent\@thmtmpitemindent}
\def\@opargbegintheorem#1#2#3{\@thmtmpitemindent\itemindent\topsep 0pt\rmfamily\trivlist
  \item[\hskip\labelsep{\indent\itshape #1\ #2\ (#3):}]\itemindent\@thmtmpitemindent}
\def\@endtheorem{\endtrivlist\unskip}

\def\QED{\mbox{\rule[0pt]{1.3ex}{1.3ex}}}

\renewcommand\refname{References}
  {\def\@noitemerr{\@latex@warning{Empty bibliography}}\endlist}

\makeatother

\DeclareCaptionLabelFormat{ieee}{#1~#2.}
\newtheorem{definition}{\bf Definition}
\newtheorem{theorem}{\bf Theorem}
\newtheorem{lemma}{\bf Lemma}
\newtheorem{proposition}{\bf Proposition}
\newtheorem{remark}{\bf Remark}

\newcommand{\eps}{\varepsilon}
\newcommand{\rann}{\operatorname{range}}

\newcommand{\enc}{\mathrm{enc}}
\begin{document}

\title{Supervisory Control under Partial Observation: Where Observation Consistency Becomes Decidable}

\author{Shaowen Miao, Jan Komenda, Tom\'{a}\v{s} Masopust, and Yiding Ji\thanks{Supported by National Natural Science Foundation of China grants 62303389 and 62373289; Guangdong Scientific Research Platform and Project Scheme grant 2024KTSCX039; Youth Science and Technology Talent Support Program of Guangdong Provincial Association for Science and Technology grant SKXRC2025463; Guangdong Provincial Key Lab of Integrated Communication, Sensing and Computation for Ubiquitous Internet of Things grant 2023B1212010007; and by RVO 67985840.
\emph{(Corresponding authors: Yiding Ji; Tom\'{a}\v{s} Masopust)}}
\thanks{S. Miao and Y. Ji are with Robotics and Autonomous Systems Thrust, The Hong Kong University of Science and Technology (Guangzhou), Guangzhou, China.
  \emph{smiao585@connect.hkust-gz.edu.cn, \linebreak jiyiding@hkust-gz.edu.cn}}
\thanks{J. Komenda is with the Institute of Mathematics of the Czech Academy of Sciences, 115 67 Prague, Czechia.
  \emph{komenda@ipm.cz}}
\thanks{T. Masopust is with the Faculty of Science, Palacky University Olomouc, Czechia. \emph{tomas.masopust@upol.cz}}
}

\maketitle

\begin{abstract}
    Observation consistency (OC) and modified observation consistency (MOC) are structural conditions used in hierarchical and modular supervisory control under partial observation. Their verification for languages generated by deterministic finite automata is PSPACE-hard, whereas decidability was open. We answer this question by showing that both problems are undecidable, that is, there are no algorithms verifying OC or MOC. On the positive side, we identify a decidable class defined by a restriction on the plant: if every cycle of the automaton contains a transition labeled by an observable high-level event, then verification of both conditions is PSPACE-complete.
\end{abstract}


\section{Introduction}
    Supervisory control of large discrete-event systems (DESs) is often limited by the size of the plant model. Hierarchical control addresses this limitation by synthesizing the supervisor on a smaller high-level abstraction of the plant, and by implementing it on the original, low-level plant. For this to work, properties established at the high level must transfer to the low level.

    Under partial observation, a specification admits a supervisor if and only if it is controllable and observable. Observability is not preserved under union, and hence a specification need not have a supremal observable sublanguage. Normality is a stronger condition that is preserved under union and therefore admits a supremal element, which makes it the property of choice in synthesis~\cite{wonham2019supervisory,cassandras2021introduction}. Transferring observability or normality between the two levels requires that the abstraction be compatible with the observations available at the low level.

    Boutin et al.~\cite{boutin2011hierarchical} introduced \emph{observation consistency} (OC) as such a compatibility condition between a plant and its abstraction, and showed that it is sufficient to preserve observability under abstraction. Informally, OC requires that every two high-level strings that the high-level observer cannot distinguish have low-level realizations that the low-level observer cannot distinguish either. Komenda and Masopust~\cite{komenda2023hierarchical} studied the analogous problem for normality and observed that OC is insufficient for that purpose: normality of the high-level supervisor transfers to the low level only if the low-level realization of one of the two strings may be prescribed in advance. They therefore introduced \emph{modified observation consistency} (MOC), which differs from OC in the order of the quantifiers: for \emph{every} low-level realization of the first string there must be a matching realization of the second one. Under MOC, the computation of supremal normal sublanguages commutes with the abstraction, which yields maximal permissiveness with respect to the low-level plant.

    Both conditions have since been used beyond their original setting. Miao et al.~\cite{miao2026hierarchical} replaced projections by the set-valued observations that arise from communication delays and losses as well as from sensor attacks, and introduced the corresponding conditions nondeterministic OC (NOC) and MNOC, thereby extending hierarchical synthesis to networked and cyber-attacked DESs. In modular systems, where the monolithic plant may be exponentially larger than each local component, MOC guarantees that the global nonblocking and maximally permissive normal supervisor is the parallel composition of the local normal supervisors~\cite{komenda2023modular}. The same condition transfers $(L,P)$-normality to local projections, which yields a modular computation of supremal $(L,P)$-normal sublanguages and modular enforcement of critical observability and non-weak opacity~\cite{miao2026modularcowo}, as well as of diagnosability and prognosability~\cite{hu2025active}.

    All of these applications share a common obstacle. Since no algorithm verifying OC/MOC was known, they resort to sufficient conditions that are easy to check; namely, alphabetic restrictions~\cite{komenda2023hierarchical,komenda2023modular}. Those conditions are strictly stronger than OC/MOC. What was known about verification is that it is PSPACE-hard~\cite{komenda2023hierarchical}, which leaves open whether a verification algorithm exists at all.

    We show that no such algorithm exists: verification of OC and of MOC is undecidable. The proofs proceed through rational relations, whose universality problem is undecidable~\cite{fischer1968multitape,muscholl2019facets}. Given an arbitrary rational relation, we construct a DFA such that the generated language is MOC (respectively OC) if and only if the relation is universal. The constructions use five events, two of them observable and three high-level, showing that undecidability does not rest on large alphabets.

    These results are negative, and they raise the question of what remains verifiable. We discuss a restriction on the plant, namely, the plants in which the two projections synchronize after a bounded number of events, and show that verification of both conditions is then PSPACE-complete.

    Section~\ref{sec:preliminaries} recalls the required notions. Section~\ref{sec:characterizations} characterizes both conditions by sets of low-level observations and by an inclusion of relations. Section~\ref{sec:verification-complexity} contains the two reductions. Section~\ref{sec:decidable} presents the decidable class.

\section{Preliminaries}\label{sec:preliminaries}
    An alphabet $\Sigma$ is a finite set of events. Let $\Sigma^*$ denote the free monoid of all finite strings over $\Sigma$, including the empty string $\eps$. A language is a subset of $\Sigma^*$. The prefix closure of a language $L$ is the set $\overline L=\{u\in\Sigma^*\mid uv\in L\text{ for some }v\in\Sigma^*\}$. A projection $\pi\colon\Sigma^*\to\Gamma^*$, for $\Gamma\subseteq\Sigma$, is the morphism for concatenation such that $\pi(a)=a$ for $a\in\Gamma$ and $\pi(a)=\eps$ otherwise. We write $\Gamma^+$ for the set of nonempty strings over~$\Gamma$.

    A \emph{(decision) problem} is a language $A$, where $x\in A$ means that $x$ is a yes-instance. The problem $A$ is \emph{decidable} if there is an algorithm that halts on every input and accepts exactly the yes-instances; otherwise, $A$ is \emph{undecidable}. A \emph{reduction} from $A$ to $B$ is a computable function $f$ such that $x\in A$ if and only if $f(x)\in B$; if $B$ is decidable and $A$ reduces to $B$, then $A$ is decidable as well. A \emph{polynomial-space} algorithm is an algorithm whose working space is bounded by a polynomial in the size of the input, and PSPACE is the class of problems admitting such an algorithm. A problem is PSPACE-hard if every problem from PSPACE reduces to it by a polynomial-time reduction. If a PSPACE-hard problem belongs to PSPACE, we talk about a PSPACE-complete problem.

    A deterministic finite automaton (DFA) is a tuple $G=(X,\Sigma,\delta,x_0,X_m)$, where $X$ is a finite set of states, $\Sigma$ is an alphabet, $x_0\in X$ is the initial state, $X_m$ is the set of marked states, and $\delta\colon X\times\Sigma\to X$ is a partial transition function, extended to strings in the usual way. The generated language of $G$ is the set $L(G)=\{s\in\Sigma^*\mid \delta(x_0,s)\in X\}$, and its marked language is the set $L_m(G)=\{s\in L(G)\mid \delta(x_0,s)\in X_m\}$. Thus, $L(G)$ is prefix-closed, and $L(G)=L_m(G)$ if every reachable state is marked. A \emph{nondeterministic finite automaton} (NFA) differs from a DFA in that $\delta$ maps a state and an event to a set of states; its accepted language is denoted by $\mathcal L(\cdot)$. 

    Let $\Sigma_o,\Sigma_{\mathrm{hi}}\subseteq\Sigma$ be the observable and high-level alphabets. We use the projections
    $P\colon\Sigma^*\to\Sigma_o^*$ and $Q\colon\Sigma^*\to\Sigma_{\mathrm{hi}}^*$. Set $I=\Sigma_o\cap\Sigma_{\mathrm{hi}}$, and let $Q_o\colon\Sigma_o^*\to I^*$ and $P_{\mathrm{hi}}\colon\Sigma_{\mathrm{hi}}^*\to I^*$ be the corresponding restrictions, see Fig.~\ref{fig01}. 
    \begin{figure}
        \centering
        \[
        \begin{CD}
            \Sigma^* @>Q>> \Sigma_{\mathrm{hi}}^* \\
            @VPVV @VVP_{\mathrm{hi}}V \\
            \Sigma_o^* @>Q_o>> I^*
        \end{CD}
        \]
        \caption{Commuting projections used throughout the paper.}
        \label{fig01}
    \end{figure}
    The projections satisfy
    \begin{equation}\label{eq:commuting-projections}
        Q_oP=P_{\mathrm{hi}}Q.
    \end{equation}

    We recall OC and MOC in the form used below~\cite{komenda2023hierarchical}.

\begin{definition}\label{def:oc}
    A prefix-closed language $L\subseteq\Sigma^*$ is \emph{observation consistent (OC)} with respect to $Q$, $P$, and $P_{\mathrm{hi}}$ if, for all $t,t'\in Q(L)$ such that $P_{\mathrm{hi}}(t)=P_{\mathrm{hi}}(t')$, there are $s,s'\in L$ such that $Q(s)=t$, $Q(s')=t'$, and $P(s)=P(s')$.
\end{definition}

\begin{definition}\label{def:moc}
    A prefix-closed language $L\subseteq\Sigma^*$ is \emph{modified observation consistent (MOC)} with respect to $Q$, $P$, and $P_{\mathrm{hi}}$ if, for every $s\in L$ and every $t'\in Q(L)$ such that $P_{\mathrm{hi}}(Q(s))=P_{\mathrm{hi}}(t')$, there is $s'\in L$ such that $Q(s')=t'$ and $P(s')=P(s)$.
\end{definition}

    The difference is the order of the quantifiers. OC may choose both low-level representatives, whereas MOC must preserve the observation of the given string $s$, see Fig.~\ref{fig:OC_MOC} for illustration.

\begin{figure}
    \begin{center}
        \includegraphics[width=\linewidth]{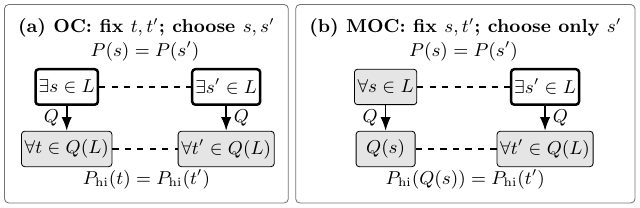}
        \caption{Comparison of OC and MOC. OC selects both realizations $s,s'$ after $t,t'$ are fixed, whereas MOC fixes $s$ in advance and selects only the matching realization $s'$.}\label{fig:OC_MOC}
    \end{center}
\end{figure}

    It is known that MOC implies OC~\cite[Lemma~10]{komenda2023hierarchical}.

\subsection{Rational Relations}
    Let $X$ and $Y$ be alphabets. A relation $R\subseteq X^*\times Y^*$ is \emph{rational} if there are a finite alphabet $C$, a regular language $H\subseteq C^*$, and morphisms $f\colon C^*\to X^*$ and $g\colon C^*\to Y^*$ s.t.
    \begin{equation}\label{eq:nivat-form}
        R=\{(f(k),g(k))\mid k\in H\} \,.
    \end{equation}

    Rational relations are rational subsets of the monoid $X^*\times Y^*$. Consequently, they are closed under union, concatenation, and Kleene star, and the Cartesian product $K\times M$ of two regular languages $K\subseteq X^*$ and $M\subseteq Y^*$ is rational~\cite[Chapter~III]{berstel1979transductions},~\cite[Chapter~IV]{sakarovitch2009elements}.

     If $X$ and $Y$ are disjoint, every rational relation $R$ admits a representation by a regular language $K_R\subseteq(X\cup Y)^*$ s.t.
    \begin{equation}\label{eq:synchronization-form}
        R=\{(\pi_X(w),\pi_Y(w))\mid w\in K_R\},
    \end{equation}
    where $\pi_X \colon (X\cup Y)^* \to X^*$ and $\pi_Y \colon (X\cup Y)^* \to Y^*$ are projections. The two representations can be effectively computed from each other~\cite[Chapter~III]{berstel1979transductions},~\cite[Chapter~IV]{sakarovitch2009elements}.

    For $R\subseteq X^*\times Y^*$, let $\rann(R)=\{y\in Y^*\mid (x,y)\in R\text{ for some }x\in X^*\}$. The relation is \emph{universal} if $R=X^*\times Y^*$. We write $A\mathbin{\dot\cup}B$ for union of disjoint sets $A$ and $B$.

    Given a rational relation by one of the representations above, it is undecidable whether the relation is universal~\cite{fischer1968multitape,muscholl2019facets}.

    The disjointness of $X$ and $Y$ required by~\eqref{eq:synchronization-form} is no restriction: renaming the alphabets by bijections preserves both rationality and universality, and hence we may always pass to disjoint renamed copies.

\section{Characterizations of the Conditions}\label{sec:characterizations}
    For every $t\in Q(L)$, define
    \begin{equation}\label{eq:observation-set}
        \mathcal O_L(t)=\{P(s)\mid s\in L,\ Q(s)=t\}
    \end{equation}
    to be the set of low-level observations of strings whose high-level projection is $t$.

\begin{proposition}\label{prop:observation-set-characterizations}
    Let $L\subseteq\Sigma^*$ be prefix-closed. Then:
    \begin{enumerate}
    \item $L$ is MOC if and only if, for all $t,t'\in Q(L)$,
    \begin{equation}\label{eq:moc-observation-sets}
        P_{\mathrm{hi}}(t)=P_{\mathrm{hi}}(t')
        \quad\Longrightarrow\quad
        \mathcal O_L(t)=\mathcal O_L(t')\,.
    \end{equation}
    \item $L$ is OC if and only if, for all $t,t'\in Q(L)$,
    \begin{equation}\label{eq:oc-observation-sets}
        P_{\mathrm{hi}}(t)=P_{\mathrm{hi}}(t')
        \quad\Longrightarrow\quad
        \mathcal O_L(t)\cap\mathcal O_L(t')\ne\emptyset.
    \end{equation}
    \end{enumerate}
\end{proposition}
\hspace{0pt}

\begin{proof}
    Assume that $L$ is MOC, and let $t,t'\in Q(L)$ have the same $P_{\mathrm{hi}}$-projection. If $u\in\mathcal O_L(t)$, there is $s\in L$ such that $Q(s)=t$ and $P(s)=u$, and MOC gives $s'\in L$ with $Q(s')=t'$ and $P(s')=u$, that is, $\mathcal O_L(t)\subseteq\mathcal O_L(t')$. Interchanging $t$ and $t'$ gives $\mathcal O_L(t) = \mathcal O_L(t')$.
    Conversely, assume~\eqref{eq:moc-observation-sets}, and let $s\in L$ and $t'\in Q(L)$ satisfy $P_{\mathrm{hi}}(Q(s))=P_{\mathrm{hi}}(t')$. Then, $P(s)\in\mathcal O_L(Q(s))=\mathcal O_L(t')$ and the definition of $\mathcal O_L(t')$ gives $s'\in L$ with $Q(s')=t'$ and $P(s')=P(s)$. Thus, $L$ is MOC.

    If $L$ is OC, then for every $t,t'\in Q(L)$ with $P_{\mathrm{hi}}(t)=P_{\mathrm{hi}}(t')$, there are $s,s'\in L$ with $Q(s)=t$, $Q(s')=t'$, and $P(s)=P(s')$, and hence $P(s) \in \mathcal O_L(t)\cap\mathcal O_L(t')$. Conversely, assume~\eqref{eq:oc-observation-sets}. If $u\in\mathcal O_L(t)\cap\mathcal O_L(t')$, there are $s,s'\in L$ with $Q(s)=t$, $Q(s')=t'$, and $P(s)=u=P(s')$, and therefore $L$ is OC.
\end{proof}

\begin{remark}\label{rem:special-case}
    If $\Sigma_o \cap \Sigma_{\mathrm{hi}} = \emptyset$, then $P_{\mathrm{hi}}(t) = P_{\mathrm{hi}}(t')$ for all $t,t'\in Q(L)$, and the conditions of Proposition~\ref{prop:observation-set-characterizations} read
    (i) $L$ is MOC if and only if $\mathcal O_L(t) = \mathcal O_L(t')$ for all $t,t' \in Q(L)$, and (ii) $L$ is OC if and only if $\mathcal O_L(t) \cap \mathcal O_L(t') \neq \emptyset$ for all $t,t' \in Q(L)$.
\end{remark}

    For a language $L\subseteq\Sigma^*$ and projections $P$ and $Q$, we define the relation
    \begin{equation}\label{eq:induced-projection-relation}
        \mathcal R_L=\{(P(s),Q(s))\mid s\in L\}
    \end{equation}
    and the set of compatible pairs as
    \[
        \mathcal C_L = \{(u,v)\in P(L)\times Q(L)\mid Q_o(u)=P_{\mathrm{hi}}(v)\}.
    \]
    If $L$ is regular, then $\mathcal R_L$ is rational by~\eqref{eq:nivat-form}. By~\eqref{eq:commuting-projections}, $\mathcal R_L\subseteq\mathcal C_L$, and we show that the other inclusion characterizes MOC.

\begin{proposition}\label{prop:moc-relational}
    A prefix-closed language $L$ is MOC if and only if $\mathcal C_L \subseteq \mathcal R_L$.
\end{proposition}
\begin{proof}
    Assume that $L$ is MOC, and let $(u,v)\in\mathcal C_L$. Choose $s\in L$ with $P(s)=u$. Then $P_{\mathrm{hi}}(Q(s)) = Q_o(P(s)) = Q_o(u) = P_{\mathrm{hi}}(v)$. By MOC, there is $s'\in L$ with $P(s')=u$ and $Q(s')=v$, that is, $(u,v)\in\mathcal R_L$, and hence $\mathcal C_L\subseteq\mathcal R_L$.

    Conversely, let $s\in L$ and $t'\in Q(L)$ with $P_{\mathrm{hi}}(Q(s))=P_{\mathrm{hi}}(t')$. Since $P_{\mathrm{hi}}(Q(s)) = Q_o(P(s)) = P_{\mathrm{hi}}(t')$, $(P(s),t')\in\mathcal C_L \subseteq \mathcal R_L$, and therefore there is $s'\in L$ such that $P(s')=P(s)$ and $Q(s')=t'$, that is, $L$ is MOC.
\end{proof}

\section{Main Results}\label{sec:verification-complexity}
    To prove the results, we state two elementary observations.

\begin{lemma}\label{lem:dfa-realization}
    Let $K\subseteq \Omega^*$ be a regular language, and let $a \notin \Omega$. Then $L = \Omega^*\cup K a$ is prefix-closed and regular.
\end{lemma}

\begin{proof}
    Every proper prefix of a string of $Ka$ belongs to $\Omega^*$; regularity follows from closure under concatenation and union. A DFA for $L$ can be constructed from the DFA for $K$ in polynomial time. Indeed, we first complete the DFA for $K$ by adding a new sink state, and then add a new state $x$ together with $a$-transitions to it from every marked state. Finally, we mark all states.
\end{proof}

    The next construction combines standard closure properties of rational relations with a block encoding. It reduces every instance of the universality problem to an instance over two-letter alphabets, which is what makes the alphabets of the reductions below fixed and small.

\begin{lemma}\label{lem:relation-extension}
    Let $R\subseteq X^*\times Y^*$ be rational, and let $A=\{a_1,a_2\}$ and $B=\{b_1,b_2\}$ be alphabets. There is an effectively constructible rational relation $S\subseteq A^*\times B^*$ such that $\rann(S)=B^*$, and $S$ is universal if and only if $R$ is universal.
\end{lemma}
\begin{proof}
    Choose $m\ge 2$ such that $2^m\ge\max\{|X|,|Y|\}$, and let $h\colon X^*\to A^*$ and $h'\colon Y^*\to B^*$ be morphisms that map every event to a string of length $m$ and are injective on events. Then $h$ and $h'$ are injective, and $h(X^*)$ and $h'(Y^*)$ are regular. Let
    \[
        D_A=a_1h(X^*),\qquad D_B=b_1h'(Y^*),
    \]
    which are regular as well, and define
    \begin{align*}
        S={}& \{(a_1 h(x),b_1 h'(y)) \in D_A \times D_B \mid (x,y)\in R\}\\
        &\cup\ (A^*\setminus D_A)\times B^*\\
        &\cup\ A^*\times(B^*\setminus D_B)\,.
    \end{align*}
    The relation $S'=\{(h(x),h'(y))\mid (x,y)\in R\}$ is rational, because writing $R$ in the form~\eqref{eq:nivat-form} exhibits $S'$ in the form~\eqref{eq:nivat-form} with the morphisms $h\circ f$ and $h'\circ g$. The first term of $S$ is the concatenation of a finite relation with $S'$; in the other two terms, both factors are regular languages, and their Cartesian product is rational. Since rational relations are closed under concatenation and union, $S$ is rational~\cite[Chapter~III]{berstel1979transductions},~\cite[Chapter~IV]{sakarovitch2009elements}.

    Every string of $D_A$ is nonempty, and therefore $\eps\notin D_A$. The second term then gives $(\eps,v)\in S$ for every $v\in B^*$, and hence $\rann(S)=B^*$.

    Finally, the last two terms of $S$ contain every pair outside $D_A\times D_B$. Every string of $D_A$ is $a_1h(x)$ for a unique $x\in X^*$, because $h$ is injective, and similarly every string of $D_B$ is $b_1h'(y)$ for a unique $y\in Y^*$. Thus, for pairs in $D_A\times D_B$,
    \[
        (a_1h(x),b_1h'(y))\in S \iff (x,y)\in R \,,
    \]
    and therefore $S$ is universal if and only if $R$ is universal.
\end{proof}

    The lemma does not require $A$ and $B$ to be disjoint. The proof of Theorem~\ref{thm:moc-undecidable} uses disjoint $A$ and $B$, since it applies~\eqref{eq:synchronization-form}, whereas the reduction of Section~\ref{sec:oc} uses $A=B$.

\subsection{Undecidability of MOC Verification}\label{sec:moc}
    We now reduce universality of rational relations to MOC verification using the two preceding constructions.

\begin{theorem}\label{thm:moc-undecidable}
    Given a DFA $G$ and alphabets $\Sigma_o,\Sigma_{\mathrm{hi}}\subseteq\Sigma$, it is undecidable whether $L(G)$ is MOC. The result holds even if $|\Sigma|=5$, $|\Sigma_o|=2$, and $|\Sigma_{\mathrm{hi}}|=3$.
\end{theorem}

\begin{proof}
    Let $R\subseteq X^*\times Y^*$ be a rational relation, let $A=\{0,1\}$ and $B=\{2,3\}$, and let $S\subseteq A^*\times B^*$ be the relation provided by Lemma~\ref{lem:relation-extension}. Since $A$ and $B$ are disjoint, \eqref{eq:synchronization-form} applies to $S$: writing $U=A\mathbin{\dot\cup}B$, there is a regular language $K\subseteq U^*$ such that
    \[
        S = \{(\pi_{A}(w),\pi_{B}(w))\mid w\in K\}.
    \]
    Take a fresh event $a$, and define
    \[
        \Sigma=U\mathbin{\dot\cup}\{a\},\qquad
        \Sigma_o=A,\qquad
        \Sigma_{\mathrm{hi}}=B\mathbin{\dot\cup}\{a\}.
    \]
    Then $\Sigma = \Sigma_o \mathbin{\dot\cup} \Sigma_{\mathrm{hi}}$ and $\Sigma_o \cap \Sigma_{\mathrm{hi}} =\emptyset$. Let
    \begin{equation}\label{eq:moc-language}
         L = U^* \cup Ka \,.
    \end{equation}
    By Lemma~\ref{lem:dfa-realization}, $L$ is prefix-closed and effectively represented by a DFA.

    Since $a\in\Sigma_{\mathrm{hi}}$, the high-level projections of strings of $Ka$ are the strings $\pi_{B}(k)a$ with $k\in K$, that is, the strings $ta$ with $t\in\rann(S)=B^*$. Together with $Q(U^*)=B^*$, this gives
    \[
        Q(L)=B^*\cup B^*a.
    \]
    Fix $t\in B^*$. For every $u \in A^{*}$, the string $ut$ belongs to $U^{*}$ and satisfies $P(ut) = u$ and $Q(ut) = t$, and hence $u \in \mathcal O_L(t)$, that is, $A^* \subseteq \mathcal O_L(t)$. Conversely, $P(U^{*}) \subseteq A^{*}$, the high-level projection of every string of $Ka$ ends with $a$, and hence differs from $t$, so the strings of $Ka$ do not contribute to $\mathcal O_L(t)$. Therefore,
    \begin{equation}\label{eq:obsmoc-free}
        \mathcal O_{L}(t) = A^{*}.
    \end{equation}
    A string of $L$ with high-level projection $ta$ is of the form $ka$ with $k\in K$ and $\pi_{B}(k)=t$, and, as $a\notin\Sigma_o$, its observation is $\pi_{A}(k)$. Hence
    \begin{equation}\label{eq:obsmoc-marked}
        \mathcal O_{L}(ta) = \{x\in A^*\mid (x,t)\in S\}.
    \end{equation}
    Since $\Sigma_o \cap \Sigma_{\mathrm{hi}} =\emptyset$, Remark~\ref{rem:special-case} says that $L$ is MOC if and only if all the sets~\eqref{eq:obsmoc-free} and~\eqref{eq:obsmoc-marked} coincide, that is, if and only if $\{x \in A^* \mid (x,t)\in S\}=A^*$ for every $t\in B^*$. The latter means that $S=A^*\times B^*$. By Lemma~\ref{lem:relation-extension},
    \[
      L \text{ is MOC}
      \iff S = A^{*} \times B^{*}
      \iff R = X^{*} \times Y^{*} .
    \]
    Hence an algorithm deciding MOC would decide universality of rational relations, which is a contradiction.
\end{proof}

\subsection{Undecidability of OC Verification}\label{sec:oc}
    Let $R_0\subseteq X_0^*\times Y_0^*$ be rational, and let 
    \[
        V=\{2,3\}\,.
    \]
    Applying Lemma~\ref{lem:relation-extension} with $A=B=V$ yields a rational relation $R\subseteq V^*\times V^*$ that is universal if and only if $R_0$ is universal. Write $R$ in the form~\eqref{eq:nivat-form}, that is,
    \[
        R=\{(f(k),g(k))\mid k\in H\},
        \qquad H\subseteq C^*\ \text{regular},
    \]
    and put
    \[
        Z=\{0,1\} \,.
    \]
    Choose $m\ge 2$ such that $2^m \ge |C|$, and let $e\colon C^* \to Z^*$ be a morphism that maps every event of $C$ to a string of length $m$ and is injective on events; then $e$ is injective, and the length of every string of $e(C^*)$ is divisible by $m$. 
    Define the morphisms
    \[
        h_f(c)=e(c) f(c),\qquad h_g(c)=e(c) g(c),\qquad c\in C,
    \]
    with codomain $(Z\cup V)^*$. The language
    \begin{align}
        J = {} & 0 2 V^* ~\cup~ 1 3 V^* ~\cup~ 2 h_f(H) ~\cup~ 3 h_g(H)
        \label{eq:oc-relation-language}
    \end{align}
    is regular. Let $\pi_Z\colon(Z\cup V)^*\to Z^*$ and $\pi_V\colon(Z\cup V)^*\to V^*$ denote the projections. Since $e(c)\in Z^*$ and $f(c),g(c)\in V^*$, we have, for every $k\in H$,
    \begin{align*}
        \pi_Z(h_f(k))&=e(k)\,, & \pi_V(h_f(k))&=f(k)\,,\\
        \pi_Z(h_g(k))&=e(k)\,, & \pi_V(h_g(k))&=g(k)\,.
    \end{align*}
    The four languages in~\eqref{eq:oc-relation-language} are pairwise disjoint, because their strings begin with $0$, $1$, $2$, and $3$, respectively. Consequently, the relation induced by the projections on $J$ consists of the following pairs:
    \begin{align}
        &(0,2 x) ~~x\in V^*\,, &
        &(1,3 y) ~~y\in V^*\,,\notag\\
        &(e(k),2 f(k)) ~~k\in H\,, &
        &(e(k),3 g(k)) ~~k\in H\,.
        \label{eq:oc-projection-relation}
    \end{align}
    Consequently, the associated observation sets are
    \begin{align}
        \mathcal O_{J}(2 x)&=\{0\}\cup e(\{k\in H\mid f(k)=x\})\,,\label{eq:alpha-observation-set}\\
        \mathcal O_{J}(3 y)&=\{1\}\cup e(\{k\in H\mid g(k)=y\})\,.
        \label{eq:beta-observation-set}
    \end{align}
    Here $\mathcal O_{J}$ is defined as in~\eqref{eq:observation-set}, using the projections $\pi_Z$ and $\pi_V$ instead of $P$ and $Q$.

    The one-letter words $0$ and $1$ are distinct, and neither belongs to $e(C^*)$, whose strings have length divisible by $m\ge2$. These two observation sets therefore intersect if and only if they contain a common string $e(k)$; since $e$ is injective, it is if and only if a $k\in H$ satisfies $f(k)=x$ and $g(k)=y$. Hence,
    \begin{equation}\label{eq:cross-observation-set}
        \mathcal O_{J}(2 x)\cap \mathcal O_{J}(3 y)\ne\emptyset
        \quad\Longleftrightarrow\quad (x,y)\in R \,.
    \end{equation}

    We are now ready to prove the main result for OC. The proof builds on the constructions developed above.
\begin{theorem}\label{thm:oc-undecidable}
    Given a DFA $G$ and alphabets $\Sigma_o,  \Sigma_{\mathrm{hi}} \subseteq \Sigma$, it is undecidable whether $L(G)$ is OC. The result holds even if $|\Sigma|=5$, $|\Sigma_o|=2$, and $|\Sigma_{\mathrm{hi}}|=3$.
\end{theorem}
\begin{proof}
    Introduce a fresh event $a$ and set
    \[
        \Sigma=Z\mathbin{\dot\cup}V\mathbin{\dot\cup}\{a\}=\{0,1,2,3,a\},~~
        \Sigma_o=Z,~~
        \Sigma_{\mathrm{hi}}=V\mathbin{\dot\cup}\{a\}.
    \]
    Then $\Sigma=\Sigma_o\mathbin{\dot\cup}\Sigma_{\mathrm{hi}}$ and $I=\Sigma_o\cap \Sigma_{\mathrm{hi}} = \emptyset$. Define
    \begin{equation}\label{eq:oc-language}
        L = (Z\cup V)^*\cup J a.
    \end{equation}
    Lemma~\ref{lem:dfa-realization} shows that $L$ is prefix-closed, regular, and effectively represented by a DFA. Since $a\in\Sigma_{\mathrm{hi}}$,
    \[
        Q(L) = V^*\cup(2 V^*\cup 3 V^*)a = V^*\cup V^+a \,.
    \]

    For every $t\in V^*$, only the language $(Z\cup V)^*$ contributes to $\mathcal O_L(t)$, because the high-level projection of every string of $J a$ ends with $a$. Since $(Z\cup V)^*$ realizes every low-level observation together with every high-level string,
    \[
        \mathcal O_{L}(t)=Z^*.
    \]
    The remaining high-level strings end with $a$ and hence arise only from $J a$. Since $a\notin\Sigma_o$, the projection $P$ erases $a$, and~\eqref{eq:alpha-observation-set} and~\eqref{eq:beta-observation-set} yield
    \begin{align*}
        \mathcal O_{L}(2 x a) &=\{0\}\cup e(\{k\in H\mid f(k)=x\}),\\
        \mathcal O_{L}(3 y a) &=\{1\}\cup e(\{k\in H\mid g(k)=y\}).
    \end{align*}
    All these sets are nonempty subsets of $Z^*$. Consequently, each of them intersects $\mathcal O_{L}(t)=Z^*$; any two sets of the first kind contain $0$, and any two sets of the second kind contain $1$. The only pairs left to consider are $2 x a$ and $3 y a$, and by~\eqref{eq:cross-observation-set}, their observation sets intersect if and only if $(x,y)\in R$. Since $2 x a\in Q(L)$ and $3 y a\in Q(L)$ for all $x,y\in V^*$, and since $I=\emptyset$, Remark~\ref{rem:special-case} gives
    \begin{align*}
        L \text{ is OC}
        &\Longleftrightarrow (x,y)\in R\text{ for all }x,y\in V^*\\
        &\Longleftrightarrow R=V^*\times V^*
        \Longleftrightarrow R_0=X_0^*\times Y_0^*\,,
    \end{align*}
    where the last equivalence is by Lemma~\ref{lem:relation-extension}. An algorithm deciding OC for DFAs would consequently decide universality of rational relations, which is impossible.
\end{proof}

\section{A Decidable Class}\label{sec:decidable}
    The undecidability proofs use that $\mathcal R_L$ is an arbitrary rational relation: the projections $P(s)$ and $Q(s)$ may drift apart without bound. We now introduce a restriction on the plant that removes this freedom and makes both conditions decidable.
  
    The events of $I=\Sigma_o \cap \Sigma_{\mathrm{hi}}$ are the only events both projections retain: they are the only available synchronization.
 
\begin{definition}\label{def:ibounded}
    A language $L\subseteq\Sigma^*$ is \emph{$I$-bounded} if there is $N\in\mathbb N$ such that no string of $L$ contains more than $N$ consecutive events of $\Sigma\setminus I$.
\end{definition}

    Let $G$ be a DFA. Then the condition is decidable in time linear in the size of $G$ by checking acyclicity of the reachable subgraph of $G$ obtained by deleting the transitions labeled by $I$. We thus have the following.
\begin{lemma}\label{lem:ibounded-check}
    Let $G$ be a DFA with all states reachable. Then $L(G)$ is $I$-bounded if and only if $G$ has no cycle all of whose transitions are labeled by events of $\Sigma\setminus I$. \hfill\QED
\end{lemma}

    In this section, we assume that all states of $G$ are reachable and write $X$ for the state set of $G$. Fix $N=|X|-1$ and put
    \[
        B_o=(\Sigma_o\setminus I)^{\le N}
        \quad \text{ and } \quad 
        B_{\mathrm{hi}}=(\Sigma_{\mathrm{hi}}\setminus I)^{\le N} \,.
    \]
    Every string $w$ over an alphabet $\Omega$ with $I\subseteq\Omega$ uniquely factors
    \[
        w=w_0\,i_1\,w_1\cdots i_n\,w_n,
        \quad i_j\in I,\ w_j\in(\Omega\setminus I)^*,
    \]
    where $i_1,\dots,i_n$ are the successive occurrences of the events of $I$ in $w$. We call it the \emph{$I$-factorization} of $w$ and $w_0,\dots,w_n$ its \emph{blocks}. In particular, every $u\in\Sigma_o^*$ has the $I$-factorization $u=p_0i_1p_1\cdots i_np_n$ with $p_j\in(\Sigma_o\setminus I)^*$, and every $v\in\Sigma_{\mathrm{hi}}^*$ has the $I$-factorization $v=q_0i'_1q_1\cdots i'_mq_m$ with $q_j\in(\Sigma_{\mathrm{hi}}\setminus I)^*$. If $Q_o(u)=P_{\mathrm{hi}}(v)$, then $n=m$ and $i_j=i'_j$ for all $j$, and the two factorizations can be interleaved. Over the finite alphabet
    \[
        \Gamma=(B_o\times B_{\mathrm{hi}})\ \dot\cup\ I \,,
    \]
    define the encoding of such a pair by
    \[
        \enc(u,v)=(p_0,q_0)\,i_1\,(p_1,q_1)\,i_2\cdots i_n\,(p_n,q_n).
    \]
    Since the $I$-factorization is unique, the map $\enc$ is injective on the pairs it is defined for. It is defined on all of $\mathcal C_L$ whenever $L$ is $I$-bounded. We write $\enc_o(u)=p_0\,i_1\,p_1\cdots i_n\,p_n$ and $\enc_{\mathrm{hi}}(v)=q_0\,i_1\,q_1\cdots i_n\,q_n$ for the corresponding encodings over the alphabets $B_o\mathbin{\dot\cup}I$ and $B_{\mathrm{hi}}\mathbin{\dot\cup}I$. 

    By construction, the two kinds of letters alternate in an encoding, which begins and ends with a letter of $B_o\times B_{\mathrm{hi}}$. In other words, every encoding belongs to the regular language
    \begin{equation}\label{eq:alternating}
        \mathrm{Alt}=(B_o\times B_{\mathrm{hi}})\bigl(I\,(B_o\times B_{\mathrm{hi}})\bigr)^*\subseteq\Gamma^*,
    \end{equation}
    and the same holds, mutatis mutandis, for $\enc_o$ and $\enc_{\mathrm{hi}}$.

\begin{theorem}\label{thm:ibounded-decidable}
    Let $G$ be a DFA such that $L(G)$ is $I$-bounded. Then it is decidable whether $L(G)$ is OC, and whether $L(G)$ is MOC. Both problems are in \textup{PSPACE}.
\end{theorem}
\begin{proof}
    Write $L=L(G)$. We show that $\enc(\mathcal R_L)$ and $\enc(\mathcal C_L)$ are regular and effectively constructible, so that
    \[
        L\ \text{is MOC}
        \iff
        \enc(\mathcal C_L)\subseteq\enc(\mathcal R_L)
    \]
    by Proposition~\ref{prop:moc-relational} and injectivity of $\enc$.
 
    For $\enc(\mathcal R_L)$, let $X$ and $x_0$ be the state set and the initial state of $G$, and build the NFA $A_{\mathcal R}$ over $\Gamma$ with the state set $X\times\{0,1\}$, the initial state $(x_0,0)$, the accepting states $X\times\{1\}$, and the transitions
    \begin{itemize}
        \item $(x,0)\xrightarrow{(p,q)}(x',1)$ if $\delta(x,w)=x'$ for $w\in(\Sigma\setminus I)^{\le N}$ with $P(w)=p$ and $Q(w)=q$, and
        \item $(x,1)\xrightarrow{i}(x',0)$ if $\delta(x,i)=x'$, for $i\in I$.
    \end{itemize}
    The second component forces the two kinds of letters to alternate, and hence $\mathcal L(A_{\mathcal R})\subseteq\mathrm{Alt}$.

    We claim that $\mathcal L(A_{\mathcal R})=\enc(\mathcal R_L)$. 
    
    Let $(u,v)\in\mathcal R_L$, say $u=P(s)$ and $v=Q(s)$ for some $s\in L$, and let $s=w_0i_1w_1\cdots i_nw_n$ be its $I$-factorization; by $I$-boundedness of $L$, $|w_j|\le N$ for every $j$. Since $s\in L$, the map $\delta$ is defined along $s$, and hence
    \begin{multline*}
        (x_0,0)\xrightarrow{(P(w_0),Q(w_0))}(y_0,1)\xrightarrow{i_1}(x_1,0) \xrightarrow{(P(w_1),Q(w_1))}\cdots \\
        \cdots\xrightarrow{(P(w_n),Q(w_n))}(y_n,1)\,,
    \end{multline*}
    where $y_j=\delta(x_j,w_j)$ and $x_j=\delta(y_{j-1},i_j)$, is a run of $A_{\mathcal R}$; it ends in $X\times\{1\}$ and is thus accepting. Since $P$ and $Q$ fix the events of $I$, and no block $w_j$ contains such an event, $P(w_0)\,i_1\cdots i_n\,P(w_n)$ and $Q(w_0)\,i_1\cdots i_n\,Q(w_n)$ are the $I$-factorizations of $u$ and $v$, respectively. The label of the run is thus $\enc(u,v)$, that is, $\enc(u,v) \in \mathcal L(A_{\mathcal R})$.
 
    Conversely, let $\gamma\in\mathcal L(A_{\mathcal R})$. Every letter of $\Gamma$ determines which of the two kinds of transitions may read it, and an accepting run starts in $(x_0,0)$ and ends in $X\times\{1\}$; hence the run of $A_{\mathcal R}$ on $\gamma$ has the form
    \[
        (x_0,0)\xrightarrow{(p_0,q_0)}(y_0,1)\xrightarrow{i_1}(x_1,0)\cdots\xrightarrow{(p_n,q_n)}(y_n,1) \,,
    \]
    and $\gamma=(p_0,q_0)\,i_1\cdots i_n\,(p_n,q_n)$. For every $j$, choose $w_j\in(\Sigma\setminus I)^{\le N}$ realizing the transition $(x_j,0)\xrightarrow{(p_j,q_j)}(y_j,1)$, that is, $\delta(x_j,w_j)=y_j$, $P(w_j)=p_j$, and $Q(w_j)=q_j$, and set $s=w_0i_1w_1\cdots i_nw_n$. Then $\delta(x_0,s)=y_n$, and hence $s\in L$, and $w_0i_1w_1\cdots i_nw_n$ is the $I$-factorization of $s$, because no $w_j$ contains an event of $I$. Consequently, $P(s)=p_0\,i_1\cdots i_n\,p_n$, $Q(s)=q_0\,i_1\cdots i_n\,q_n$, and $\gamma=\enc(P(s),Q(s))\in\enc(\mathcal R_L)$.

    For $\enc(\mathcal C_L)$, the same construction applied to $P$ alone, respectively $Q$ alone, gives NFAs over $B_o\dot\cup I$ and $B_{\mathrm{hi}}\dot\cup I$ accepting the block encodings of $P(L)$ and of $Q(L)$. Let $\rho_o\colon\Gamma^*\to(B_o\dot\cup I)^*$ and $\rho_{\mathrm{hi}}\colon\Gamma^*\to(B_{\mathrm{hi}}\dot\cup I)^*$ be the morphisms erasing the second, respectively the first, component of the pairs and fixing $I$. Then
    \[
        \enc(\mathcal C_L)=\rho_o^{-1}(\enc_o(P(L)))\cap\rho_{\mathrm{hi}}^{-1}(\enc_{\mathrm{hi}}(Q(L))),
    \]
    because a pair lies in $\mathcal C_L$ if and only if its components lie in $P(L)$ and $Q(L)$ and their $I$-projections agree, the latter being enforced by the common $I$-letters of the encoding. Regular languages are closed under inverse morphisms and intersection, and hence $\enc(\mathcal C_L)$ is regular, which completes the proof for MOC.
 
    To consider the verification of OC, put
    \[
        \Gamma_3=(B_o\times B_{\mathrm{hi}}\times B_{\mathrm{hi}})\,\dot\cup\,I
        \quad\text{and}\quad
        \Gamma_2=(B_{\mathrm{hi}}\times B_{\mathrm{hi}})\,\dot\cup\,I,
    \]
    and extend $\enc$ to triples and to pairs of high-level strings as before: if all blocks are of length at most $N$ and $Q_o(u)=P_{\mathrm{hi}}(v)=P_{\mathrm{hi}}(v')$, then $\enc(u,v,v')\in\Gamma_3^*$ and $\enc(v,v')\in\Gamma_2^*$ are the interleavings of the respective $I$-factorizations.

    Let $A_3$ be the product of two copies of $A_{\mathcal R}$ that synchronizes on the first component of the pairs and on the letters of $I$. The two copies read letters of the same kind at the same time, and hence their alternation components stay equal; the product thus has $2|X|^2$ reachable states and accepts only alternating words. Since $\mathcal L(A_{\mathcal R})=\enc(\mathcal R_L)$,
    \[
        \mathcal L(A_3)=\{\enc(u,v,v')\mid u\in\mathcal O_L(v)\cap\mathcal O_L(v')\} \,,
    \]
    the synchronization on the first component expressing that $(u,v)$ and $(u,v')$ lie in $\mathcal R_L$ for the same $u$.

    Set $D=\{\enc(v,v')\mid v,v'\in Q(L),\ P_{\mathrm{hi}}(v)=P_{\mathrm{hi}}(v')\}$ and let $\sigma_1,\sigma_2\colon\Gamma_2^*\to(B_{\mathrm{hi}}\dot\cup I)^*$ be the morphisms erasing the second, respectively the first, component of the pairs and fixing $I$. We show that
    \[
        D=\sigma_1^{-1}(\enc_{\mathrm{hi}}(Q(L)))\cap\sigma_2^{-1}(\enc_{\mathrm{hi}}(Q(L))) \,,
    \]
    which is the counterpart of the identity established for $\enc(\mathcal C_L)$. Both $\sigma_1$ and $\sigma_2$ are length-preserving and send letters of $B_{\mathrm{hi}}\times B_{\mathrm{hi}}$ to letters of $B_{\mathrm{hi}}$ and letters of $I$ to themselves; since every encoding is alternating, so is every word of either preimage. Let $\gamma=(q_0,q_0')\,i_1\cdots i_n\,(q_n,q_n')$ be such a word and put $v=q_0i_1\cdots i_nq_n$ and $v'=q_0'i_1\cdots i_nq_n'$. Then $\sigma_1(\gamma)=\enc_{\mathrm{hi}}(v)$ and $\sigma_2(\gamma)=\enc_{\mathrm{hi}}(v')$, that is, $\gamma$ lies in the intersection if and only if $v,v'\in Q(L)$, by injectivity of $\enc_{\mathrm{hi}}$. Moreover, the two strings share the letters $i_1,\dots,i_n$, which are their $I$-projections; the interleaving therefore enforces $P_{\mathrm{hi}}(v)=P_{\mathrm{hi}}(v')$, and $\gamma=\enc(v,v')$. Conversely, every element of $D$ arises in this way. Hence $D$ is regular, being accepted by the product of two copies of the NFA for $\enc_{\mathrm{hi}}(Q(L))$ in which a letter $(q,q')$ makes the first copy read $q$ and the second copy read $q'$, and a letter of $I$ is read by both; this product has $O(|X|^2)$ states.

    Finally, let $\pi\colon\Gamma_3^*\to\Gamma_2^*$ erase the first component of the triples and fix $I$. Then $\pi(\enc(u,v,v'))=\enc(v,v')$, and
    \[
        \pi(\mathcal L(A_3))=\{\enc(v,v')\mid\mathcal O_L(v)\cap\mathcal O_L(v')\ne\emptyset\} \,,
    \]
    where the pairs $(v,v')$ range over $Q(L)\times Q(L)$ with $P_{\mathrm{hi}}(v)=P_{\mathrm{hi}}(v')$, since a nonempty $\mathcal O_L(v)\cap\mathcal O_L(v')$ forces both conditions. Being length-preserving, $\pi$ requires no closure argument beyond renaming the letters: $\pi(\mathcal L(A_3))$ is accepted by the NFA obtained from $A_3$ by replacing the label $(p,q,q')$ by $(q,q')$, on the same state set. By Proposition~\ref{prop:observation-set-characterizations},
    \[
        L\ \text{is OC}
        \iff
        D\subseteq\pi(\mathcal L(A_3)) \,,
    \]
    an inclusion between regular languages, and hence decidable.
 
    \emph{Space complexity.} Note that $\Gamma$ has exponentially many letters, and that none of the automata is constructed explicitly; they are represented succinctly, and their transitions are computed on the fly. This is possible because $N<|X|$, so that a letter of $\Gamma$ is a pair of strings of length at most $|X|$ and is written in polynomial space, and because the transition relations of all the automata are decided in polynomial time as follows.
    
    For $A_{\mathcal R}$, whether $(x,0)\xrightarrow{(p,q)}(x',1)$ is a transition amounts to reachability from $(x,\eps,\eps)$ to $(x',p,q)$ in the product of $G$ with the prefixes of $p$ and the prefixes of $q$, where the transitions labeled by events of $I$ are removed. Since $L$ is $I$-bounded, every $w\in(\Sigma\setminus I)^*$ for which $\delta(x,w)$ is defined satisfies $|w|\le N$. The automata for $\enc_o(P(L))$, for $\enc_{\mathrm{hi}}(Q(L))$, and hence for $\enc(\mathcal C_L)$ and for $D$, are handled in the same way.

    The automaton accepting $\pi(\mathcal L(A_3))$ needs a separate argument, because relabeling $(p,q,q')$ by $(q,q')$ quantifies the first component existentially over the exponentially many letters of $B_o$. Its transitions are nevertheless decided in polynomial time: a letter $(q,q')$ leads from $(x,\tilde x)$ to $(y,\tilde y)$ if and only if $(x,\tilde x,\eps,\eps)$ reaches $(y,\tilde y,q,q')$ in the product of two copies of $G$ with the prefixes of $q$ and $q'$, in which the two copies move simultaneously and on the same event along the events of $\Sigma_o\setminus I$, which are the events recorded by the erased component, and independently along the remaining events of $\Sigma\setminus I$. The common first component $p$ is then the sequence of events read simultaneously, and it never has to be written down.

    Both criteria are thus inclusions between languages of succinctly represented NFAs with $O(|X|^2)$ states, which are decided in polynomial space by the standard on-the-fly algorithm~\cite{stockmeyer1973word,savitch1970relationships}. The exponential size of $\Gamma$ is therefore immaterial; the complexity comes from the inclusion test.
\end{proof}
 
    The class is nontrivial in both directions. Let $\Sigma_o=\{a,c\}$ and $\Sigma_{\mathrm{hi}}=\{x,c\}$, so that $I=\{c\}$. The prefix closure $L$ of $(ac+xc)^*$ is $I$-bounded with $N=1$, and it is neither OC nor MOC: the high-level strings $c$ and $xc$ satisfy $P_{\mathrm{hi}}(c)=c=P_{\mathrm{hi}}(xc)$, since $x\notin I$, while $\mathcal O_L(c)=\{ac,aca\}$ and $\mathcal O_L(xc)=\{c,ca\}$ are disjoint, since every string of the former begins with $a$ and every string of the latter with $c$; the failure of MOC then follows because MOC implies OC. Over the same alphabets, the language $c^*$ is $I$-bounded and MOC. Every finite language is $I$-bounded. In contrast, the languages built in Theorems~\ref{thm:moc-undecidable} and~\ref{thm:oc-undecidable} contain $U^*$, respectively $(Z\cup V)^*$, with $I=\emptyset$, and are not $I$-bounded.
 
    The complexity bound of Theorem~\ref{thm:ibounded-decidable} cannot be improved, since it is matched by the following lower bound. Its instances have $N=3$, and hence a block alphabet $\Gamma$ of polynomial size; the exponential size of $\Gamma$ in general is therefore not what makes the problems hard.

\begin{theorem}\label{thm:ibounded-hard}
    Verification of OC, and verification of MOC, is PSPACE-hard on the $I$-bounded instances, even if $N=3$.
\end{theorem}
\begin{proof}
    We reduce from the problem whether, given DFAs $A_1,\dots,A_n$ over an alphabet $\Delta$, $\bigcup_{i=1}^{n}L_m(A_i)=\Delta^*$, which is PSPACE-complete~\cite{kozen1977lower}.

    For MOC, take fresh pairwise distinct events $o$, $h$, and $\tau_1,\dots,\tau_n$, and set
    \[
        \Sigma=\Delta\mathbin{\dot\cup}\{o,h\}\mathbin{\dot\cup}\{\tau_1,\dots,\tau_n\},~
        \Sigma_o=\Delta\mathbin{\dot\cup}\{o\},~
        \Sigma_{\mathrm{hi}}=\Delta\mathbin{\dot\cup}\{h\},
    \]
    so that $I=\Delta$ and the events $\tau_i$ lie in neither $\Sigma_o$ nor $\Sigma_{\mathrm{hi}}$ and are therefore erased by both projections. Let
    \[
        L=\Delta^*\cup\Delta^*o\cup\Delta^*h\cup\bigcup_{i=1}^{n}\tau_i\bigl(\Delta^*\cup L_m(A_i)o\cup L_m(A_i)oh\bigr).
    \]
    This language is prefix-closed, and a DFA for it is obtained from a state $m$ carrying a self-loop on every event of $\Delta$, an $o$-transition and an $h$-transition from $m$ to two states without successors, a $\tau_i$-transition from $m$ to the initial state of a copy of $A_i$, an $o$-transition from every marked state of that copy to a state $r_i$, and an $h$-transition from $r_i$ to a state without successors. Since every $A_i$ is complete, that is, $L(A_i)=\Delta^*$, $\tau_i \Delta^*$ is obtained for free. The DFA has $\sum_i|A_i|+O(n)$ states and is computed in polynomial time.

    Since $o\notin\Sigma_{\mathrm{hi}}$, $h\notin\Sigma_o$, and the $\tau_i$ are erased, we obtain $P(L)=\Delta^*\cup\Delta^*o$ and $Q(L)=\Delta^*\cup\Delta^*h$, and $Q_o$ and $P_{\mathrm{hi}}$ erase $o$ and $h$, respectively. Hence
    \[
        \mathcal C_L=\{(w,w),(w,wh),(wo,w),(wo,wh)\mid w\in\Delta^*\} \,.
    \]
    The strings $w$, $wo$, and $wh$ of $L$ realize the first three families of pairs. A string of $L$ whose observation contains $o$ and whose high-level projection contains $h$ must be of the form $\tau_i w o h$ with $w\in L_m(A_i)$, and it realizes $(wo,wh)$. By Proposition~\ref{prop:moc-relational},
    \[
        L\ \text{is MOC}\iff \mathcal C_L = \mathcal R_L \iff \bigcup\nolimits_{i}L_m(A_i)=\Delta^* \,.
    \]
    Finally, the events outside $I=\Delta$ are $o$, $h$, and the $\tau_i$, and no string of $L$ contains more than three of them consecutively, as in $\tau_ioh$, so $L$ is $I$-bounded with $N=3$.

    For OC, take fresh pairwise distinct events $o_1,o_2,c,h_1,h_2,\tau_1,\dots,\tau_n$, and set $\Sigma_o=\Delta\mathbin{\dot\cup}\{o_1,o_2,c\}$ and $\Sigma_{\mathrm{hi}}=\Delta\mathbin{\dot\cup}\{h_1,h_2\}$, again with $I=\Delta$ and with the $\tau_i$ erased by both projections. Let $L$ be the prefix closure of
    \[
        \Delta^*o_1h_1\cup\Delta^*o_2h_2\cup\bigcup_{i=1}^{n}\tau_iL_m(A_i)c\,\{h_1,h_2\} \,,
    \]
    which is again realized by a DFA of polynomial size. Then $Q(L)=\Delta^*\cup\Delta^*\{h_1,h_2\}$ and $P_{\mathrm{hi}}$ erases $h_1$ and $h_2$. Two strings of $Q(L)$ have the same $P_{\mathrm{hi}}$-projection if and only if each of them is $w$ or $wh_k$ for $w\in\Delta^*$, thus the pairs to be considered are $(w,w)$, $(w,wh_k)$, $(wh_k,wh_k)$, and $(wh_1,wh_2)$.

    A string of $L$ with high-level projection $wh_k$ is either $wo_kh_k$ or of the form $\tau_iwch_k$ with $w\in L_m(A_i)$, and hence
    \[
        \mathcal O_L(wh_k)=
        \begin{cases}
            \{wo_k,wc\}, & \text{if } w\in\bigcup\nolimits_i L_m(A_i),\\
            \{wo_k\}, & \text{otherwise},
        \end{cases}
    \]
    for $k=1,2$, whereas $\mathcal O_L(w)$ contains both $wo_1$ and $wo_2$, as witnessed by the strings $wo_1$ and $wo_2$ of $L$. The first three kinds of pairs are satisfied by the observation $wo_k$, which is in $\mathcal O_L(w)$ and in $\mathcal O_L(wh_k)$. For the pairs $(wh_1,wh_2)$, the sets $\mathcal O_L(wh_1)$ and $\mathcal O_L(wh_2)$ intersect if and only if $w\in\bigcup_i L_m(A_i)$, since $o_1\ne o_2$ and $wc$ belongs to both in that case. By Proposition~\ref{prop:observation-set-characterizations}, $L$ is OC if and only if $\bigcup_i L_m(A_i) = \Delta^*$. Again, no string of $L$ contains more than three consecutive events outside $\Delta$, as in $\tau_ich_1$, and hence $N=3$.
\end{proof}

    Consequently, both problems are PSPACE-complete. 

\section{Conclusion}\label{sec:conclusion}
    Verification of OC and MOC is undecidable for prefix-closed regular languages represented by DFAs. On the positive side, restricting the plant yields a decidable class. If the automaton has no cycle labeled by events outside the shared alphabet, the two projections synchronize after a bounded number of events, both induced relations become regular over a finite alphabet of blocks, and OC and MOC become PSPACE-complete. The source of undecidability is thus not the alphabet size but the possibility of unbounded desynchronization between the two projections.

\section*{Acknowledgment}
    Claude (Anthropic) and Codex (OpenAI) independently suggested the use of rational relations in Theorems~\ref{thm:moc-undecidable} and~\ref{thm:oc-undecidable}. The authors take full responsibility for all statements and proofs.

\bibliographystyle{IEEEtran}
\bibliography{references}

\end{document}